\documentclass[%
 pra,twocolumn,
 superscriptaddress,
 amsmath,amssymb,
 aps, showpacs, 
floatfix,
dvipsnames,
]{revtex4-2}

\usepackage{optidef}
\usepackage{braket}
\usepackage{float}
\usepackage{hyperref}
\usepackage{ragged2e}

\usepackage{amsthm,color}

\newtheorem{remark}{Remark}

\newtheorem{proposition}{Proposition}
\usepackage{orcidlink}

\usepackage{amsmath}
\usepackage{graphicx}
\usepackage{subcaption}  % Include figure files
\usepackage{dcolumn}% Align table columns on decimal point
\usepackage{bm}% bold math

\begin{document}

\title{Unitary gate synthesis via polynomial optimization}

\author{Llorenç Balada Gaggioli \orcidlink{0009-0006-3930-3457}}
\email{llorenc.balada.gaggioli@fel.cvut.cz}
\affiliation{%
Department of Computer Science, 
Czech Technical University in Prague, Prague 2, Czech Republic}
\affiliation{%
LAAS-CNRS, Université de Toulouse, France
}%

\author{Denys I. Bondar\orcidlink{0000-0002-3626-4804}}
 \email{dbondar@tulane.edu}
\affiliation{%
 Department of Physics and Engineering Physics, Tulane University,  New Orleans, LA 70118, USA
}

\author{Jiri Vala\orcidlink{0000-0001-7795-8602}}
\email{jiri.vala@mu.ie}
\affiliation{Department of Physics, Maynooth University, Maynooth, Ireland}
\affiliation{School of Theoretical Physics, Dublin Institute for Advanced Studies, Dublin, Ireland}
\affiliation{Tyndall National Institute, Cork, Ireland}

\author{Roman Ovsiannikov}
\email{roman.ovsiannikov@kipt.kharkiv.ua}
\affiliation{NSC Kharkiv Institute of Physics and Technology, 61108 Kharkiv, Ukraine}

\author{Jakub Marecek\orcidlink{0000-0003-0839-0691}}
\email{jakub.marecek@fel.cvut.cz}
\affiliation{%
Department of Computer Science, 
Czech Technical University in Prague, Prague 2, Czech Republic}

\date{\today}

\begin{abstract}

Quantum optimal control plays a crucial role in the development of quantum technologies, particularly in the design and implementation of fast and accurate gates for quantum computing. Here, we present a method to  synthesize gates using the Magnus expansion. In particular, we formulate a polynomial optimization problem that allows us to find the global solution without resorting to approximations of the exponential. The global method we use provides a certificate of globality and lets us do single-shot optimization, which implies it is generally faster than local methods.
 By optimizing over Hermitian matrices generating the unitaries, instead of the unitaries themselves, we can reduce the size of the polynomial to optimize, leading to fast convergence and scalability.
 %The discrete Magnus expansion is used for  piecewise constant controls for the presented optimal control problem. 
 Numerical experiments comparing our results with CRAB and GRAPE show that we maintain high accuracy while providing globality certificates. 
\end{abstract} 

\maketitle

\section*{Introduction}

The steering of quantum dynamical systems toward desired target states and operators is one of the principal concerns of quantum optimal control \cite{Koch_2022,Morzhin_2019,Zhang_2017,Glaser_2015,Dong_2010,Brif_2010,DAlessandro2007-ht}. Among its applications, a crucial one is the synthesis of unitary gates, which form the computational gates in quantum computing circuits \cite{Nielsen_Chuang_2010}. Implementing unitary operations with high fidelity, minimal time, and constrained \cite{PhysRevLett.89.188301,Nielsen_2006,roque2020engineeringfasthighfidelityquantum} by hardware specifications is fundamental for the development of quantum technologies such as  quantum computing, quantum simulations, and quantum sensing \cite{Ac_n_2018}. However, achieving high precision and robustness is still an obstacle in the improvement of quantum computing in the Noisy Intermediate-Scale Quantum (NISQ) era \cite{Preskill2018quantumcomputingin}.

Optimizing these unitary gates can be done in different ways. Many currently used methods are local solvers, which can result in very good fidelities. However, we are not assured to have found the best possible solution, the global optimum, as they can get stuck in local optima. The use of global optimization methods allows us to find the global solution in a single run, instead of using many random initializations, as local solvers do. This leads not only to a generally faster way of extracting the solution, but also to a more accurate one. This is crucial to the implementation of quantum technologies, for example, in quantum circuits formed by a large number of unitaries, where the slight reduction of the error for each gate can lead to a significant improvement of the total accuracy.

Motivated by the desire to obtain global solutions to constrained quantum optimal control problems, in \cite{bondar2025globallyoptimalquantumcontrol} QCPOP is presented, where quantum optimal control problems are reformulated as polynomial optimization problems. This allows us to find the global solution to the problem using a hierarchy of semidefinite programs \cite{lasserre}. QCPOP is based on solving the Schrödinger equation approximately via the Magnus expansion \cite{Magnus,marecek2020quantumoptimalcontrolmagnus}, to handle the time dependence of the control Hamiltonian, and Chebyshev polynomials \cite{10.1063/1.448136} to approximate the exponential function.

Here we consider the unitary gate synthesis problem with the goal of formulating a polynomial optimization problem with which we can find the global solution of the problem efficiently and for a wider range of controls. We show that by optimizing the Hermitian matrices that generate the unitaries we are concerned about, instead of the unitaries themselves, we can get equally accurate results for much less runtime to solve the problem. Firstly, there is less computation time required to obtain the polynomial to optimize, and secondly this polynomial is of less degree, and therefore also faster to minimize.

We also consider piecewise constant controls, and present an alternative method, what we call the discrete Magnus expansion, to the generally used to iterative application of time independent unitaries for each time step. The improvement on efficiency of the proposed method is illustrated by comparing its running time to the running time of QCPOP. And, to show how we maintain high accuracy, we compare both the continuous and piecewise constant controls to the established methods GRAPE \cite{GRAPE} and CRAB \cite{CRAB}.

\section*{Gate synthesis}

Consider a finite-dimensional quantum system that evolves following a unitary operator $U(t)$, which satisfies the Schrödinger equation
\begin{equation}
    \frac{\partial}{\partial t} U(t) = -\frac{i}{\hbar}H(t)U(t),
    \label{Schro}
\end{equation}
where $U(0)=I$. We set $\hbar=1$ and we consider the Hamiltonian of the system to be of the following form
\begin{equation}
    H(t)=H_0+E(t)H_c,
\end{equation}
where $H_0$ is the free Hamiltonian, describing the intrinsic dynamics of the system, and $H_c$ is the control Hamiltonian, which regulates the influence of the control function, $E(t)$, on the system.

In most quantum control problems, we want to optimize the shape of the control pulse that will drive the system to the desired state or unitary. Most control pulses have a shape that can be well approximated via a polynomial, so we will take the control function to be
\begin{equation}
    E(t)=\sum_{k=0}^{m-1} x_k t^{k},
\end{equation}
and $\{x_0,x_1,\hdots,x_{m-1}\}$ are the parameters that we want to optimize. We could also use a linear combination of functions of time weighted by the coefficients $\mathbf{x}$, as long as these functions of time are integrable.

To study how the system evolves in time with the control function, we have to solve Equation ~\ref{Schro}. For this, we can use the Magnus series \cite{Magnus}, which lets us write the unitary evolution of the system at some time $T$ as
\begin{align}\label{eq:MagnusApprox}
    U(T) = \lim_{n \rightarrow \infty} \exp [ \Lambda_n ], 
    \qquad \Lambda_n = \sum_{k=1}^{n} \Omega_k . 
\end{align}

If we let $A(t)=-iH(t)$, using $B(s,t) \equiv [A(s), A(t)]$, $C(s,t,u) \equiv [A(s), B(t,u)]$ the first three terms of the Magnus expansion are as follows. 
\begin{align}
    \Omega_1 &= \int_0^T \!\! A(t) \, dt  , \label{eq:Omega1} \;\;\; \Omega_2 = \frac{1}{2} \int_0^T \!\! \int_0^{t} \!\! B(t,s) \, dt  ds , \\
    \Omega_3 &= \frac{1}{3!} \int_0^T \!\! \int_0^{t} \!\! \int_0^{s} \!\!  C(t,s,u)  -  C(u,t,s)  \, dt ds du. \label{eq:Omega3}
\end{align}
%The brackets $[\cdot,\cdot]$ denote the commutator of the underlying Lie algebra of the one-parameter family of operators $\{A(t_i)\}$, $i=1,2,\ldots$.  
The Magnus series converges so long as $ 
    \int _{0}^{T}\| A(t) \|_2 \ dt < \pi$ \cite{Blanes_2009}. 
An important feature of the Magnus expansion is that $\exp [ \Lambda_n ]$ is unitary for every $n$. Also, note that for our system, $H(t)$, the series has a the polynomial nature, making it very compatible with polynomial optimization.

The Magnus expansion provides a systematic way to write the unitary evolution of the system, and by truncating it to the first $n$ terms, namely $\Lambda_n$, we have a good approximation while simultaneously gaining computational efficiency. We note that the time variable will vanish from the integrals, and what will remain is a matrix with polynomial entries in $m$ variables $\{x_0,\hdots,x_{m-1}\}$ and of order $n$. Therefore, we will write the truncated Magnus series as $\Lambda_n(T,\mathbf{x})$.

For some quantum optimal control problems, we are interested in generating a specific unitary operator, such as a gate for a quantum computer. We want to optimize the control pulse to drive the system to the target unitary $U^{\star}$. We can formulate this problem as the minimization of the infidelity between the target unitary and the unitary evolving over time $T$
\begin{equation}
    \min_{\mathbf{x}} \quad 1- \frac{1}{d^2}\text{Tr}( U^{\star \dag}U(T,\mathbf{x})),
    \label{Gate}
\end{equation}
where $U(T)$ satisfies Equation ~\ref{Schro}, and $d$ is the dimension of the Hilbert space. Both unitaries can be written as the exponential of anti-Hermitian matrices, $U(T)=e^{\Lambda}$ and $U^{\star}=e^{\Theta}$. 

\begin{proposition}
    Let $X_U=\arg\min_{\mathbf{x}} D(U(T,\mathbf{x}),U^{\star})$ be the set of global minimisers with global optimum of 0, for some distance $D$, a unitary matrix $U(T,\mathbf{x})=e^{\Lambda}$ solving $\partial_tU(t)=A(t,\mathbf{x})U(t)$, and target unitary $U^{\star}=e^\Theta$, where $\|\Theta\|_2<\pi$; and let $X_H=\arg\min_{\mathbf{x}} D(\Lambda(T,\mathbf{x}),\Theta)$ be the set of global minimisers with global optimum of 0. If $\mathbf{x^{\star}}\in X_H$, then $\mathbf{x^{\star}}\in X_U$; and if $\mathbf{\hat{x}}\in X_U$ satisfies $ \int _{0}^{T}\| A(t,\mathbf{\hat{x}}) \|_2\ dt < \pi$ then $\mathbf{\hat{x}}\in X_H$.
    \label{prop}
\end{proposition}

\begin{proof}
    If $\mathbf{x^{\star}}\in X_H$ then $D(\Lambda(T,\mathbf{x^{\star}}),\Theta) =0$, so $\Lambda(T,\mathbf{x^{\star}})=\Theta$, which implies $D(e^{\Lambda(T,\mathbf{x^{\star}})},e^\Theta)=0$, and therefore $\mathbf{x^{\star}}\in X_U$. On the other hand, if $\mathbf{\hat{x}}\in X_U$ satisfies $ \int _{0}^{T}\| A(t,\mathbf{\hat{x}}) \|_2 \ dt < \pi$ then $\Lambda(T,\mathbf{\hat{x})}$ lies in the principal branch \cite{Moan_2007}. We assumed that the generator $\Theta$ had to be in the principal branch too and, from $D(e^{\Lambda(T,\mathbf{\hat{x}})},e^\Theta)=0$, we see the generators have to be the same. From this we conclude $D(\Lambda(T,\mathbf{\hat{x}}),\Theta) =0$, hence $\mathbf{\hat{x}}\in X_H$.
\end{proof}

This proposition concerns the full Magnus expansion, but in practice we use the truncated series at some order $n$, where $\Lambda_n = \sum_{k=1}^{n} \Omega_k $. Therefore, we can define the remainder between the infinite series and the truncated one, by letting $R_n = \Lambda - \Lambda_n$, and we can consider the following.

\begin{remark}
    If we let $s(\mathbf{x})=\int _{0}^{T}\| A(t,\mathbf{x}) \|_2\ dt < \pi$, then the remainder $\|R_n\|_2$ tends to 0 and we can pick an $n$ such that $\|R_n\|_2 < \pi - s(\mathbf{x})$. It follows from the triangle inequality that $\|\Lambda_n\|_2<\pi$, so the truncated series lies in the principal branch.
\end{remark}

During this work we will assume that the order of the Magnus approximation we pick is large enough to stay in the principal branch. This is generally the case for small systems with weak drives, if we were to use this for bigger systems the order of approximation we take should be considered more carefully.

From this proposition it follows that, under the assumption of reachability of the target unitary and considering the principal branch of the logarithm, we can obtain the global minimizer by considering the anti-Hermitian generators instead of the unitaries, resulting in the optimization problem
\begin{equation}
    \min_{\mathbf{x}} \quad \|\Lambda_n(T,\mathbf{x})- \Theta \|^{2}_F .
    \label{final_problem}
\end{equation}

The function to minimize is a polynomial in variables $\{x_0,x_1,\hdots,x_{m-1}\}$ and of degree $2n$. We can minimize this polynomial globally using the moment-SOS method \cite{lasserre,Parrilo2003-qg,9fef820b69d243f2a501e933b30bd977}. This technique transforms the polynomial optimization problem into a hierarchy of semidefinite programming problems, which are globally solvable, allowing us to find the true solution to the relaxed quantum optimal control problem.

With this method we reduce the order of the polynomial compared to additionally approximating the exponential of the Magnus expansion, from $2np$ to $2n$, where the $p$ is the order of approximation of some method we may use to approximate the exponential function, like Taylor or Chebyshev polynomials. This allows us to either increase the accuracy of the control pulse, increase the truncation order of the Magnus expansion, or increase the order of the Moment-SOS hierarchy. These will lead to better accuracy in the final results and faster convergence to the global solution. We illustrate the time required to both obtain the polynomial, and minimize it, in Fig.~\ref{fig:times} for general QCPOP and for the gate synthesis method.

In Fig.~\ref{fig:times} we consider the Ising model for $N$ qubits, where $H(t)=-J\sum_{i=1}^{N-1}\sigma_i^z\sigma_{i+1}^z+E(t)\sum_{i=1}^{N}\sigma_i^x$, and we compute the time it takes to obtain the polynomial we want to optimize and the time it takes to find the global minimizer. We do this with the QCPOP method, using a 3rd order Magnus expansion and a 4th order Chebyshev approximation (to approximate the exponential), and with gate synthesis via Hermitians, for a 3rd order Magnus expansion. From the results, we conclude that by using only the Magnus expansion we have an exponentially faster time to compute the objective polynomial and a faster computation time to solve the polynomial optimization problem.

\begin{figure}[H]
    \centering
    \includegraphics[width=1\linewidth]{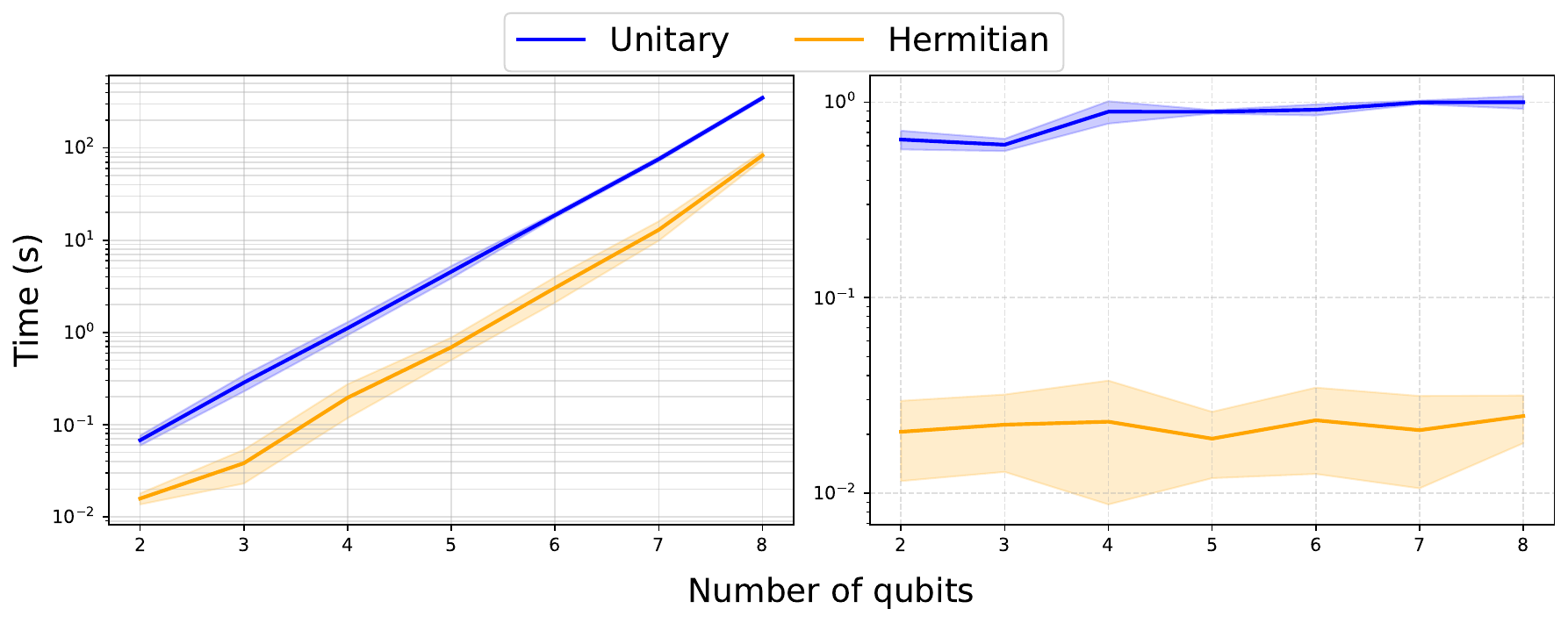}
    \caption{\textbf{Time efficiency.} We study the time to evaluate the objective (left) and time it takes to solve the polynomial optimization problem (right) for both for unitary optimization and the more efficient gate synthesis with Hermitians. Based on 5 repetitions, solid lines present the means and shaded regions present $\pm$ one standard deviation.}
    \label{fig:times}
\end{figure}

Even if our method ensures global optimality, its scalability is constrained both by the size of the physical system and the accuracy of the approximations we make. For a system of $N$ qubits, the Hamiltonian will have a size of $2^N\times 2^N$, making the computations for large systems computationally costly. However, one of the main advantages of this method is that the size of the resulting polynomial to optimize is independent of the system's size. 

Methods like GRAPE \cite{GRAPE} and CRAB \cite{CRAB} navigate an optimization landscape that is shaped by constraints (time, energy), controllability and system symmetries. For unconstrained, controllable and small systems, we generally have trap-free landscapes \cite{russell2016quantumcontrollandscapestrap, Riviello_2015} so methods like GRAPE are quite successful. However, for constrained or many-body systems, where symmetries (translational, rotational) might arise, the presence of local traps or regions with plateaus may make landscape navigation ill-conditioned and sensitive to initialization. With the method we develop in this paper we avoid these issues by approximating the landscape to a polynomial, and finding its global optimum in a single shot.

\begin{figure}[H]
    \centering
    \includegraphics[width=1\linewidth]{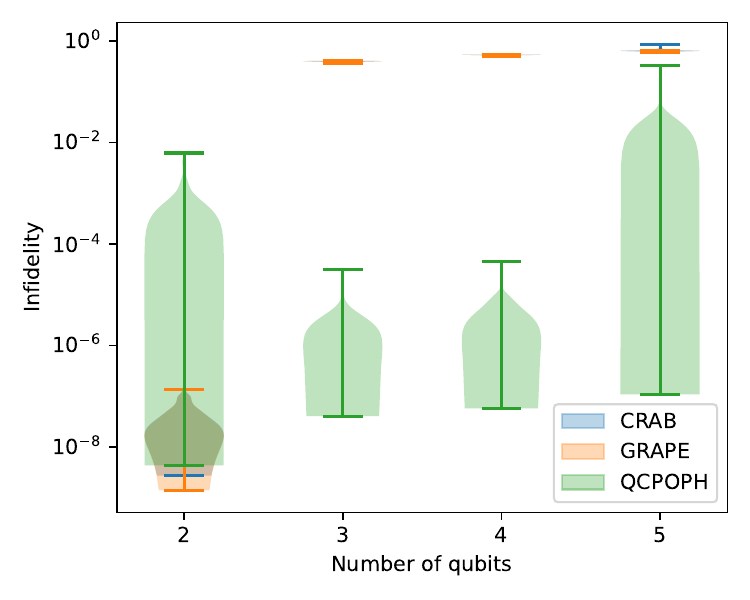}
    \caption{\textbf{Accuracy of quantum optimal control methods with number of qubits.} Infidelity comparison between different quantum optimal control methods for the XX chain with controlled global X field  for different number of qubits, tested with 100 random reachable unitary targets. The color shading is the range of infidelities and the width is the frequency. For 3 and 4 qubits, CRAB and GRAPE almost fully intersect, making them hard to distinguish.}
    \label{fig:qoc_compare}
\end{figure}

The results in Figure \ref{fig:qoc_compare} indicate that going to even larger systems while maintaining accuracy might be possible with the polynomial optimization framework. In \cite{gaggioli2025timeevolutioncontrolledmanybody} an approach using tensor networks together with polynomial optimization is presented to make the method scale well with the system's size. In this work we focus on how to make the polynomial optimization problem more scalable for more complicated control functions (with more variables), and for better approximations of the Magnus expansion.

\section*{Piecewise constant control}

If we have a piecewise constant operator $H(t)$, we can consider time steps of $\Delta t=\frac{T}{m}$, for total time period $T$ divided into $m$ segments, each of constant control. Let $A(t_i)=-i(H_0+x_iH_c)$ for unknown control parameters  $\{x_1,x_2,\hdots,x_{m}\}$, we can formulate the time evolved solution to Equation ~\ref{Schro} as
\begin{equation}
    U(T)=e^{\Delta t A(t_m)}e^{\Delta t A(t_{m-1})}\hdots e^{\Delta t A(t_2)}e^{\Delta t A(t_1)}.
\end{equation}

Motivated by the single exponential expression of the Magnus expansion, and in order to apply Equation ~\ref{final_problem} we also want to express $U(T)=e^\Sigma$ as a single exponential for the piecewise constant case. For this, we have to find an expression for $\Sigma$, we let $A_i\equiv A_i(x_i)$ and then we have 
    \begin{equation}
\begin{aligned}
    \Sigma&=\Delta t \sum_{i=1}^m A_i +\frac{(\Delta t)^2}{2}\sum_{j<i}^m[A_i,A_j]\\
    & \hspace{10pt}+\frac{(\Delta t)^3}{6}\sum_{k\leq j\leq i}^m\bigg([A_i,[A_j,A_{k}]]+[A_k,[A_j,A_i]]\bigg)+\hdots
\end{aligned}\label{BCH_extended}
\end{equation}
which we can see as a discrete Magnus expansion \cite{doikou2024quantumgroupsdiscretemagnus}. More details on this formula can be found in Appendix \ref{ap-BCH}.

Current methods to simulate piecewise constant evolution generally use iterative approaches, taking $m-1$ products of the exponentials. This approach is effective when working with numerical entries, but becomes intractable when the entries are polynomial. The polynomial entries will increase in size very fast and this can make the subsequent optimization protocol not work anymore. We therefore use the gBCHD formula for $\Sigma$, which yields low-order polynomials, facilitating the use of polynomial optimization protocols.

We now formulate the gate synthesis optimization problem as follows.
\begin{equation}
    \min_{\mathbf{x}} \quad \|\Sigma_n(T,\mathbf{x})- \Theta \|^{2}_F,
\end{equation}
where $n$ indicates the order of truncation of the discrete Magnus series.

This will result in a polynomial optimization problem of $m$ variables and order $2n$. We can solve this with the standard polynomial optimization methods outlined previously.

\section*{Example}

Let us look at an example of the scaled versions of the Hamiltonians for IBM
Q 2-qubit devices, we truncate the numerical values to 4 decimals for simplicity,
\begin{equation}
    H_0 =\begin{pmatrix}
        0 & 0 & 0 \\
        0 & 0.5159 & 0 \\
        0 & 0 & 1
    \end{pmatrix} , \quad H_c=\begin{pmatrix}
        0 & 0.7071 & 0 \\
        0.7071 & 0 & 1 \\
        0 & 1 & 0
    \end{pmatrix}.
\end{equation}

We will use a quadratic polynomial ($m=3$) as a control
and a third-order Magnus expansion ($n=3$). We want to solve the following gate synthesis optimization problem \ref{final_problem}. The details of the numerical implementation can be found in Appendix \ref{ap-numerics}.

\begin{figure}[H]
    \centering
    \includegraphics[width=1\linewidth]{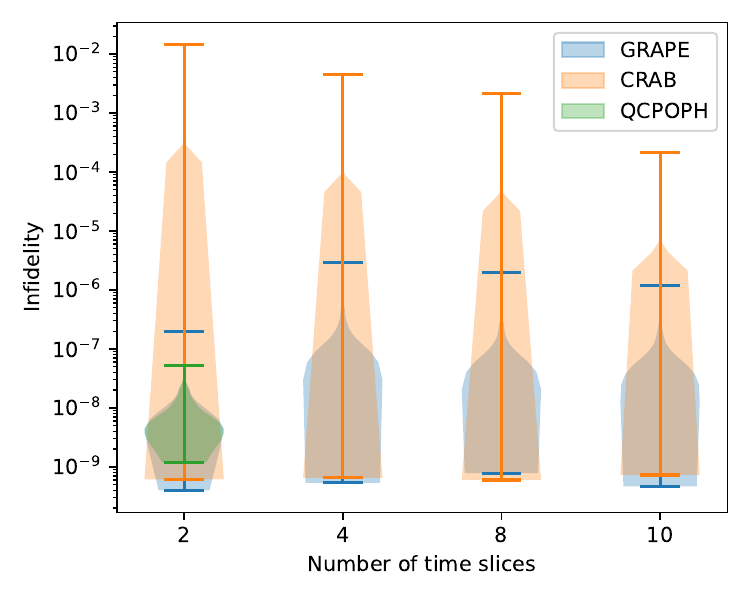}
    \caption{\textbf{Accuracy of quantum optimal control methods.} Fidelity comparison between GRAPE, CRAB and the gate synthesis with Hermitians. We test the fidelity for a sample of 1000 random values of the control function, using an order 3 Magnus expansion.}
    \label{fig:grape}
\end{figure}

In Fig. ~\ref{fig:grape} we take the integrals of the control function in the Magnus expansion, so we only compare the results with the samples of one time division. The infidelity error range for our method lies between $10^{-7}$ and $10^{-9}$, with an average infidelity similar to that of GRAPE. We improve significantly with respect to CRAB. The difference of performance between GRAPE and CRAB comes from the different parametrization of the control they use. CRAB optimizes in a truncated random Fourier basis which, in some cases, might be misaligned with the true reachable set, producing basis-induced local traps or plateaus. On the other hand, GRAPE parametrizes the full control function by discretizing it, allowing us to find a good local optimum. 

We repeat the numerical experiment for a piecewise constant control, in order to also compare this with GRAPE and CRAB. We plot the results in Fig. ~\ref{fig:piece}. Our method maintains an accuracy between $10^{-6}$ and $10^{-10}$, improving with the increasing number of slices. In this case, the more intervals we use, the more variables the final polynomial will have, and therefore the more difficult it will be to solve the polynomial optimization problem.

\begin{figure}[H]
    \centering
    \includegraphics[width=1\linewidth]{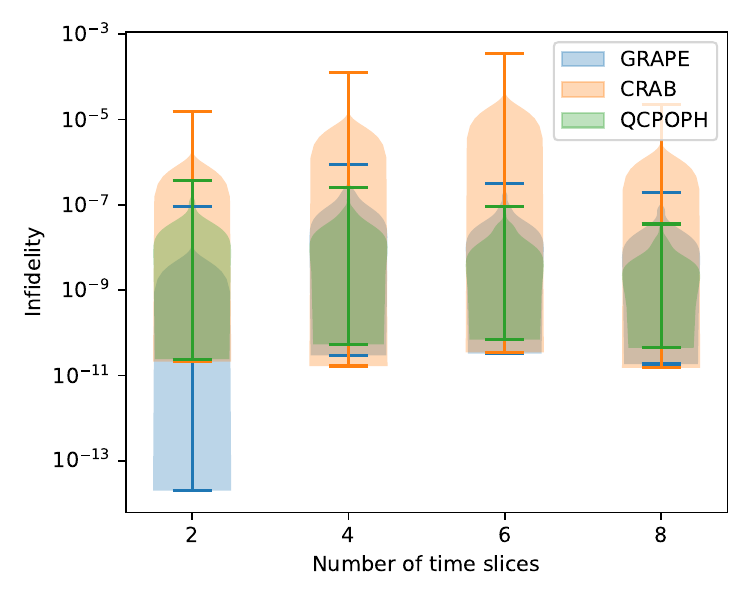}
    \caption{\textbf{Accuracy of quantum optimal control methods.} Fidelity comparison between GRAPE, CRAB and the gate synthesis with Hermitians for a piecewise constant control. We test the fidelity for a sample of 100 random values of the control function, using an order 4 discrete Magnus expansion.}
    \label{fig:piece}
\end{figure}

\section*{Conclusion}

Here we present a method to synthesise unitaries by formulating the quantum optimal control problem as a polynomial optimization problem. We exploit the structure of unitary matrices to consider only its Hermitian generators, and this leads to a major improvement in computational efficiency while maintaining a very high accuracy. By using less time to find the solution and by the nature of the polynomial optimization solver, moment-SOS, that guarantees the solution to be global, we can reduce the time required to implement a given unitary compared to local quantum optimal control protocols that need to be run many times. This, in turn, can improve the time required to implement quantum circuits with many gates significantly.

\section*{Acknowledgments}
L.B.G has been supported by European Union’s HORIZON–MSCA-2023-DN-JD programme under under the Horizon Europe (HORIZON) Marie Skłodowska-Curie Actions, grant agreement 101120296 (TENORS). D.I.B. has been supported by the Army Research Office (ARO) (grant W911NF-23-1-0288, program manager Dr.~James Joseph). The views and conclusions contained in this document are those of the authors and should not be interpreted as representing the official policies, either expressed or implied, of ARO or the U.S. Government. The U.S. Government is authorized to reproduce and distribute reprints for Government purposes notwithstanding any copyright notation herein. J.M. has been supported by the Czech Science Foundation (23-07947S).

\section*{Code availability}
All codes used in this work can be found in \textit{https://github.com/llorebaga/UGS}

\bibliography{main}

\appendix
\section{Generalized BCH formula}\label{ap-BCH}
Recall that the product of exponential functions for the scalars $a,b\in\mathbb{C}$ is $e^a e^b=e^c$, with $c=a+b$. However, when working with matrices we need a different formula to account for the non-commutativity of the objects. For matrices $A,B$ we have
\begin{equation}
    e^A e^B = e^C
\end{equation}
and the Baker-Campbell-Hausdorff (BCH) formula \cite{BCH} gives an expression for $C$:
\begin{equation}
    C=A+B+\frac{1}{2}[A,B] + \frac{1}{12}([A,[A,B]]+[B,[B,A]])+\hdots
    \label{BCH}
\end{equation}

We now want to consider an extended BCH formula that applies for multiple exponentials of matrices. For this we use the generalized Baker–Campbell–Hausdorff–Dynkin (gBCHD) formula \cite{SAENZ2002357}.

We let $A_i, i=1,\hdots,m$ be $n\times n$ matrices. If we have 
    \begin{center}
        $e^{A_m}e^{A_{m-1}}\hdots e^{A_{2}}e^{A_{1}}=e^\Sigma$
    \end{center}
    then $\Sigma$ is given by 
    \begin{equation}
\begin{aligned}
    \Sigma&=\sum_{i=1}^m A_i +\frac{1}{2}\sum_{j<i}^m[A_i,A_j]\\
    & \hspace{10pt}+\frac{1}{6}\sum_{k\leq j\leq i}^m\bigg([A_i,[A_j,A_{k}]]+[A_k,[A_j,A_i]]\bigg)+\hdots
\end{aligned}
\end{equation}

This is clearly an equivalent expression to the Magnus expansion for the discrete case, so we can think of it as the discrete Magnus expansion \cite{doikou2024quantumgroupsdiscretemagnus}.

\section{Numerical implementation}\label{ap-numerics}

To find the optimal controls of our problem, we have to compute the Magnus expansion, which can be found by computing the integral equations ~\ref{eq:Omega1} and ~\ref{eq:Omega3}, and results in a matrix with polynomial entries. What remains is to consider a unitary target and take the logarithm in the principal branch. This will then result in the polynomial to optimize.

We assume that the target unitary is reachable with our Hamiltonian. We impose this by selecting random control parameters $\mathbf{x^{\star}}$ from a
uniform distribution on the interval $[-1, +1]$, leading to a random Hamiltonian. Then we solve Equation ~\ref{Schro}, for an evolution time $T=0.5$, to obtain the target unitary $U^{\star}$. Note that we restrict the values of the control and the time interval in order to make sure that the generating anti-Hermitian satisfies $\|\Theta\|_2 <\pi$, and therefore we can make use of Proposition \ref{prop}. The objective is to recover these random control parameters, or those that also generate $U^{\star}$ with minimal error, considering that the solution might not be unique. What remains is to compute the logarithm to find the polynomial we want to optimize. 

Then, we use the Julia package \textit{TSSOS} \cite{TSSOS} to solve the polynomial optimization problem globally. We repeat this process for 1000 different random controls and target unitaries, and calculate the infidelities. We plot the results in Fig. ~\ref{fig:grape}.

We use the Python library \textit{QuTiP} \cite{qutip} to benchmark the polynomial optimization method, we call QCPOPH, with the popular methods GRAPE and CRAB, which are based on time discretisation and therefore we compute them for different number of time slices, and using a random function as the initial guess.

\end{document}